\documentclass[sigconf,balance=false]{acmart}
\AtBeginDocument{%
  }

\usepackage{algorithm}
\usepackage{algpseudocode}
\usepackage{amsmath}
\algrenewcommand\algorithmicrequire{\textbf{Input:}}
\algrenewcommand\algorithmicensure{\textbf{Output:}}
\newcommand{\qvec}{\mathbf{q}}
\newcommand{\dvec}{\mathbf{d}}
\usepackage{array}
\usepackage{cleveref}
\theoremstyle{acmplain}
\newtheorem{proposition}{Proposition}

\setcopyright{acmlicensed}
\copyrightyear{2027}
\acmYear{2027}
\acmDOI{XXXXXXX.XXXXXXX}
\acmConference[xxxx'27]{Make sure to enter the correct
  conference title from your rights confirmation email}{February xx--xx,
  2027}{New York}

\renewcommand\footnotetextcopyrightpermission[1]{}

\begin{document}

\title{AdaWidth: Query-Adaptive Embedding Width for Dense Retrieval}

\author{Shubing Yang}
\email{sueyoung@uw.edu}
\affiliation{%
  \institution{University of Washington}
  \city{Tacoma}
  \state{Washington}
  \country{USA}
}

\author{Dongfang Zhao}
\email{dzhao@cs.washington.edu}
\affiliation{%
  \institution{University of Washington}
  \city{Tacoma}
  \state{Washington}
  \country{USA}
}

\renewcommand{\shortauthors}{Yang et al.}

\begin{abstract}
High-dimensional embeddings are central to dense retrieval, but not all of these dimensions need to be evaluated at retrieval time. Existing methods reduce dimensions in two ways: truncating the same leading dimensions for every query, or masking a different subset for each query while still storing and accessing the full embedding. Yet queries within a single task differ widely in the number of dimensions they need for their rankings to stabilize. 

We introduce AdaWidth, which adapts the number of evaluated dimensions to each query within a shared prefix representation. An orthogonal prefix adapter applies a single learned rotation to queries and documents alike, concentrating discriminative signal in leading coordinates while leaving every full width inner product unchanged. A lightweight router then reads order statistics off the ranking a query has already produced, and evaluates more dimensions only for the queries whose top results would change. We further derive a prefix sufficiency analysis showing that the required number of dimensions is set by the competing documents at the retrieval cutoff: it grows logarithmically with corpus size, decreases logarithmically with retrieval depth, and remains heavy-tailed across queries. Across six retrieval tasks and five frozen encoders, AdaWidth matches the NDCG@10 of state-of-the-art dimensionality reduction using 55\% to 84\% fewer dimensions per query.
\end{abstract}

\begin{CCSXML}
<ccs2012>
 <concept>
  <concept_id>00000000.0000000.0000000</concept_id>
  <concept_desc>Do Not Use This Code, Generate the Correct Terms for Your Paper</concept_desc>
  <concept_significance>500</concept_significance>
 </concept>
 <concept>
  <concept_id>00000000.00000000.00000000</concept_id>
  <concept_desc>Do Not Use This Code, Generate the Correct Terms for Your Paper</concept_desc>
  <concept_significance>300</concept_significance>
 </concept>
 <concept>
  <concept_id>00000000.00000000.00000000</concept_id>
  <concept_desc>Do Not Use This Code, Generate the Correct Terms for Your Paper</concept_desc>
  <concept_significance>100</concept_significance>
 </concept>
 <concept>
  <concept_id>00000000.00000000.00000000</concept_id>
  <concept_desc>Do Not Use This Code, Generate the Correct Terms for Your Paper</concept_desc>
  <concept_significance>100</concept_significance>
 </concept>
</ccs2012>
\end{CCSXML}

% \ccsdesc[500]{Information retrieval}

% \keywords{Do, Not, Use, This, Code, Put, the, Correct, Terms, for,

\settopmatter{printacmref=false,printccs=true,printfolios=true}

\maketitle

\renewcommand\thefootnote{}\footnote{\noindent
The source code is available at \url{https://github.com/suey3141/AdaWidth.git}.
}\addtocounter{footnote}{-1}

\section{Introduction}

\begin{figure}[t]
  \centering
  \includegraphics[width=0.45\textwidth]{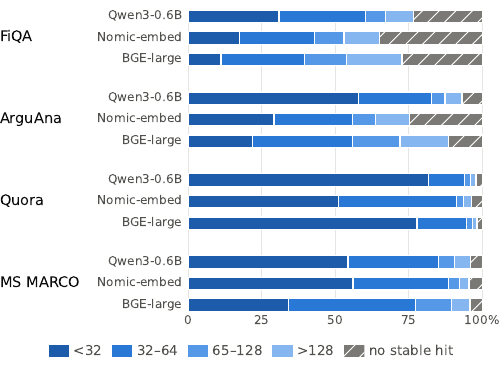}
  \caption{Distribution of stable prefix widths across four retrieval tasks and three encoders. Each bar groups queries by the smallest prefix width at which their ranking stabilizes; no stable hit indicates no stabilization within the widest prefix evaluated.}
  \label{fig:stable-width}
\end{figure}

High-dimensional embeddings are central to dense retrieval. They are the retrieval
representation in web and product
search~\cite{huang2020embedding, nigam2019semantic, li2021embedding}, open-domain question answering~\cite{karpukhin2020dense}, recommendation~\cite{covington2016deep}, and retrieval-augmented generation~\cite{lewis2020retrieval, gao2023survey}. In these systems, large corpora are searched in real time~\cite{malkov2018efficient, johnson2019billion, guo2020accelerating, yang2026}, and at that scale embedding width stops being a modelling detail and becomes a running cost. Not all of those dimensions, however, need to be evaluated at retrieval time~\cite{ethayarajh2019contextual, timkey2021rogue, mu2018allbutthetop,
wang2023dimensionality, kataiwa2025intrinsic, raunak2019effective, takeshita2025randomly,
inkiriwang2025dimensions, zhang2024evaluating}.

Existing work reduces dimensions in two ways. The first builds a representation that stays usable when truncated, so that a prefix of the embedding can be scored on its own~\cite{kusupati2022matryoshka, yoon2024search, yoon2024matryoshka, zhang2025smec, zhao2026dive, li20242d, zhuang2024starbucks}. The number of dimensions is then a property of the representation, fixed for the workload or for a target compression level, and every query is served at the same one. The second line estimates, for an individual query, which coordinates carry signal, and keeps those with the highest scores~\cite{faggioli2024dime, derasmo2024eclipse, faggioli2025codime, derasmo2026rdime, wu2026learning}. These methods do adapt to the query, but the coordinates they retain differ from query to query and are scattered across the vector: the index must hold every document at full width, and scoring gathers those coordinates out of a full width record.

However, queries differ substantially in the number of dimensions needed for their rankings to stabilize. Figure~\ref{fig:stable-width} illustrates this variation across retrieval tasks and encoders, including encoders trained with Matryoshka objectives~\cite{kusupati2022matryoshka, nussbaum2024nomic, qwen3technicalreport}. Some queries reach stable rankings with short prefixes, whereas others require much wider representations. A single width must therefore accommodate the most demanding queries, causing queries that stabilize earlier to be scored over unnecessary coordinates.

We introduce AdaWidth, a retrieval method that adapts the evaluated prefix width to each query within a shared representation. AdaWidth keeps the encoder frozen and adds two components between the encoder and the index: an orthogonal prefix adapter that makes leading coordinates more informative, and a lightweight router that determines how many of those coordinates to evaluate for each query. First, the adapter applies the same learned rotation to queries and documents, implemented as a product of Householder reflections~\cite{mhammedi2017efficient}. This rotation concentrates discriminative information in leading coordinates without changing similarities at the original width, enabling retrieval with short contiguous prefixes. Second, the router extracts features based on order statistics from the ranking produced at a fixed initial width and evaluates wider prefixes only for queries whose top results appear unstable. The adapter therefore makes short prefixes more effective, and the router converts this prefix structure into computational savings for individual queries.

We further derive a prefix sufficiency analysis of how many dimensions a query requires. A relevant document falls out of the top-$k$ at a given width exactly when $k$ non-relevant documents overtake it there. The required number of dimensions is therefore determined by an order statistic of the competition at the retrieval cutoff, not by the reconstruction accuracy of any individual similarity. This criterion differs from the reconstruction and distance preservation objectives underlying classical reduction~\cite{dasgupta2003elementary}. Under an exponential decay of the per-competitor overtaking probability, the required number grows logarithmically with corpus size, decreases logarithmically with retrieval depth, and remains heavy-tailed across queries, which is the pattern Figure~\ref{fig:stable-width} exhibits. The analysis also bounds what any orthogonal adapter can contribute: since such a map fixes every full width inner product, it can only move discriminative signal between the evaluated prefix and the discarded tail, never create it. Training the adapter is therefore energy compaction of a fixed quantity, and its effect enters the analysis as a higher rate at which discriminative information accumulates in the prefix.

We evaluate on four text retrieval tasks from BEIR~\cite{thakur2021beir}, including FiQA~\cite{maia201818}, ArguAna~\cite{wachsmuth2018retrieval}, Quora~\cite{thakur2021beir}, and MS MARCO~\cite{nguyen2016ms}, and two knowledge VQA retrieval tasks, OK-VQA~\cite{marino2019ok} and A-OKVQA~\cite{schwenk2022okvqa}. We use five frozen encoders with 137M to 7B parameters and embedding dimensions from 768 to 4096~\cite{nussbaum2024nomic, qwen3technicalreport, bge_embedding, e5Mistral}. We compare against prefix truncation, Matryoshka Adaptor~\cite{yoon2024matryoshka}, SMEC~\cite{zhang2025smec} and Learning-to-Select~\cite{wu2026learning},
each retrained on every encoder and dataset. The predicted scaling form holds on all six tasks with the predicted signs on corpus size and retrieval depth. On cost against quality, AdaWidth is ahead of all three baselines at nearly every shared operating point, and matches their NDCG@10 using 55\% to 84\% fewer dimensions per query.

This paper makes the following contributions:
\begin{itemize}
    \item We propose AdaWidth, a retrieval method that adapts the evaluated prefix width to each query within a shared representation. An orthogonal prefix adapter concentrates discriminative signal in leading coordinates. A lightweight router uses features based on order statistics from the initial ranking to evaluate wider prefixes only for queries whose top results appear unstable. (\Cref{sec:AdaWidth})
    \item We derive a prefix sufficiency analysis that characterizes the width a query requires as an order statistic of the competition at the retrieval cutoff, and predicts that it grows logarithmically with corpus size and decreases logarithmically with retrieval depth. (\Cref{sec:prefix_sufficiency_bound})
    \item We evaluate AdaWidth on six retrieval tasks with five frozen text and vision language encoders, showing that AdaWidth matches the NDCG@10 of state-of-the-art dimensionality reduction using $55\%$--$84\%$ fewer dimensions per query. (\Cref{sec:eval})
\end{itemize}

\section{Related Work}
\label{sec:RelatedWork}

\subsection{Embedding Width and Prefix Structured Representations}

Dense retrieval maps a query $q$ and each document $d$ to vectors in $\mathbb{R}^{D}$ and ranks documents by inner product or cosine similarity \cite{karpukhin2020dense}. For a candidate set $\mathcal{C}(q)$, score computation requires $\Theta(|\mathcal{C}(q)|D)$ arithmetic operations. Scoring only $m<D$ coordinates reduces this cost to $\Theta(|\mathcal{C}(q)|m)$, but may change the ordering around the Top~$k$ cutoff. Classical dimensionality reduction constructs one global lower dimensional space. The Johnson--Lindenstrauss lemma, for example, relates the dimensionality needed to preserve pairwise distances to the number of points \cite{dasgupta2003elementary}. Learned and post-hoc methods similarly construct a fixed representation for a chosen target width. ConAE learns a low dimensional bottleneck for dense retrieval \cite{liu2022dimension}, while PCA and related projections can reduce pretrained sentence embeddings without retraining \cite{zhang2024evaluating}. Empirical studies also show that the useful dimensionality is often below the native embedding dimension and that moderate coordinate removal causes limited retrieval loss \cite{wang2023dimensionality,takeshita2025randomly}. More recent work finds that compression behavior depends on corpus complexity \cite{caspari2026corect}, while Spectral Tempering adapts a global spectral transformation to each target width \cite{li2026spectral}. These methods do not characterize the width required for an individual query or the order statistic that determines entry into the Top~$k$ results.

Matryoshka Representation Learning organizes information at multiple granularities within one embedding \cite{kusupati2022matryoshka}. Given prescribed widths $\mathcal{M}=\{m_1,\ldots,m_L\}$, it jointly optimizes the leading $m$ coordinates for every $m\in\mathcal{M}$, producing nested representations that share one coordinate order. Search Adaptor learns a global transformation of frozen embeddings for retrieval at a target width \cite{yoon2024search}, while Matryoshka Adaptor extends this idea to several nested output widths through supervised or unsupervised tuning \cite{yoon2024matryoshka}. SMEC improves the optimization of nested representations through sequential compression \cite{zhang2025smec}, and DIVE focuses on stable supervised compression when relevance labels are limited \cite{zhao2026dive}. AdANNS applies Matryoshka representations of different capacities to different stages of an approximate nearest neighbor pipeline \cite{rege2023adanns}. These methods make smaller representations usable, but their widths are selected for the workload, target compression level, or retrieval stage rather than inferred from the observed ranking of each query.

\begin{figure*}[t]
  \centering
  \includegraphics[width=\textwidth]{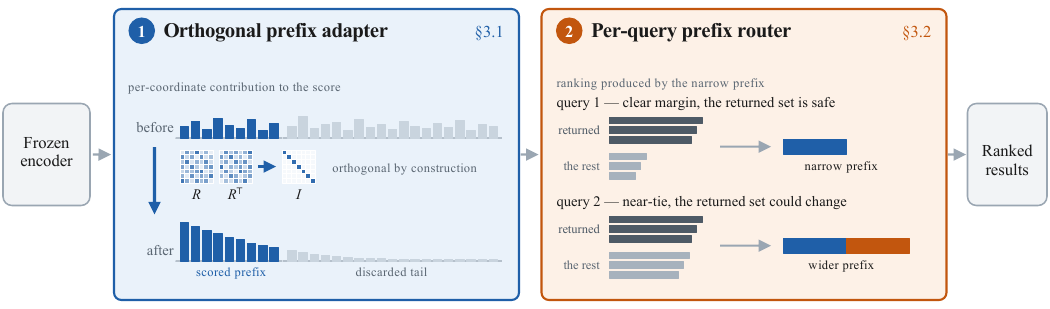}
  \caption{Overview of AdaWidth. \textbf{(1)} The orthogonal prefix
adapter rotates queries and documents by a single learned $R$; since
$\langle Rq, Rd\rangle = \langle q, d\rangle$, training can only move
per coordinate contribution out of the discarded tail and into the scored
prefix. \textbf{(2)} The prefix router reads each query's stage-1 ranking at the narrow prefix $\ell$ and rescores at the wider prefix $h$ only for the fraction $\rho$ of queries whose returned documents are not clearly separated from the rest.}
  \label{fig:overview}
\end{figure*}

\subsection{Query Specific Coordinate Selection}

A separate line of work estimates which coordinates are useful for an individual query. DIME constructs a query dependent importance vector, typically from the elementwise interaction between the query and a pseudo relevant representation, and retains the coordinates with the highest scores \cite{faggioli2024dime}. ECLIPSE incorporates pseudo irrelevant feedback, while CoDIME estimates coordinate importance from counterfactual click signals \cite{derasmo2024eclipse,faggioli2025codime}. These approaches improve the estimation of coordinate importance, but the retained dimensionality is generally selected through a global search.

RDIME replaces this search with a statistical risk criterion that can determine a different coordinate set for each query \cite{derasmo2026rdime}, while Learning to Select distills supervised coordinate importance into a predictor that operates directly on the query embedding \cite{wu2026learning}. These methods produce query specific and generally noncontiguous coordinate subsets. Although RDIME also allows the number of selected coordinates to vary, its criterion does not relate that number to corpus size or retrieval depth. AdaWidth instead learns one prefix order shared by queries and documents and selects only how much of that prefix to evaluate. Its orthogonal adapter preserves every full width inner product, while its router allocates width from the ranking already computed for each query. The prefix sufficiency analysis further connects this width to corpus size, retrieval depth, and the rate at which discriminative information accumulates in the prefix.

\subsection{Query Adaptive Retrieval Computation}

Query performance prediction estimates retrieval difficulty from the query or its resulting score distribution \cite{cronentownsend2002predicting,cummins2014document,faggioli2023query}. Cascade ranking uses early ranking evidence to decide which candidates receive more expensive processing \cite{wang2011cascade}, while related methods select retrieval configurations or evaluation depths separately for each query \cite{deveaud2019learning,ganguly2023query}. These methods adapt the ranker, candidate set, or retrieval depth. AdaWidth instead keeps the shortlist and scoring function fixed and determines whether the same candidates require a wider embedding prefix.

\section{AdaWidth}
\label{sec:AdaWidth}

AdaWidth leaves the encoder frozen and adds two components between it and the
index.  As shown in Figure~\ref{fig:overview}, an \emph{orthogonal prefix adapter}
(Section~\ref{sec:adapter}) applies a single learned rotation to queries and
documents alike, moving discriminative signal toward the leading coordinates
without altering any full-width inner product.  A \emph{prefix router}
(Section~\ref{sec:router}) then decides, per query, how many of those coordinates
to score against.  The prefix sufficiency analysis in Section~\ref{sec:prefix_sufficiency_bound}
provides a theoretical justification for both components.

\paragraph{Notation.}
Let \(E\) be a frozen encoder of width \(D\), and write
\(\qvec=E(q)\), \(\dvec=E(d)\in\mathbb{R}^{D}\) for unit-normalized embeddings.
Throughout the analysis and the evaluation we call \(D\) the \emph{encoder
width}, use \emph{width} for the number of evaluated coordinates, and reserve
\emph{coordinate} for an individual axis.
For a width \(m\le D\) let
\(\pi_m(\mathbf{x})=\mathbf{x}_{1:m}/\lVert\mathbf{x}_{1:m}\rVert_2\), and for an
orthogonal \(R\in\mathbb{R}^{D\times D}\) define the prefix score
\begin{equation}
    s^R_m(q,d)=\big\langle \pi_m(R\qvec),\,\pi_m(R\dvec)\big\rangle .
    \label{eq:prefix_score}
\end{equation}
Every method we compare is scored on a shared candidate shortlist
\(\mathcal{C}(q)\), \(\lvert\mathcal{C}(q)\rvert=500\), mined once from the
\emph{frozen} encoder at \(64\) coordinates, so that differences between methods
reflect re-ranking within a fixed pool.  The
retrieval metric is NDCG@10 with the ideal gain computed over all judged
relevant documents in the corpus.

\subsection{Orthogonal Prefix Adapter}
\label{sec:adapter}
AdaWidth learns its prefix order through an orthogonal map. Products of Householder reflections provide an exact parameterization of orthogonal matrices \cite{mhammedi2017efficient}, and the compact WY representation evaluates such products using dense matrix operations \cite{schreiber1989storage}. AdaWidth uses this parameterization to move discriminative signal into leading coordinates without changing any full width inner product.

\subsubsection{Parameterization}
We realize \(R^{\top}\) as a product of \(F\) Householder reflections,
\begin{equation}
    R^{\top}=H(\mathbf{v}_1)\,H(\mathbf{v}_2)\cdots H(\mathbf{v}_F),
    \qquad
    H(\mathbf{v})=I-2\mathbf{v}\mathbf{v}^\top,
    \label{eq:householder_product}
\end{equation}
with unit vectors \(\mathbf{v}_j\in\mathbb{R}^{D}\) obtained by normalizing free
parameters.  Orthogonality is therefore structural: it
holds at every point of training, at any learning rate, and requires no
projection step.  Evaluating \eqref{eq:householder_product} as written costs
\(F\) sequential rank-one updates, which serializes the backward pass and makes
\(F\) expensive to raise.  We instead use the compact WY representation
\begin{equation}
    H(\mathbf{v}_1)\cdots H(\mathbf{v}_F)=I-VTV^{\top},
    \label{eq:wy_form}
\end{equation}
where \(V=[\mathbf{v}_1,\dots,\mathbf{v}_F]\in\mathbb{R}^{D\times F}\) and
\(T\in\mathbb{R}^{F\times F}\) is upper triangular.  Because the reflection
vectors are unit-norm, every diagonal entry of \(T\) equals \(2\) and the
remaining entries follow the one-column recursion in
Algorithm~\ref{alg:wy}.  A batch \(X\in\mathbb{R}^{n\times D}\) is then
transformed by three dense products,
\(XR^{\top} = X-\big((XV)T\big)V^{\top}\), in \(O(nFD+nF^{2})\) time with no
sequential dependence, which is what makes \(F=64\) practical where the
sequential form is not; each row of \(XR^{\top}\) is \(R\) applied to that
embedding as in \eqref{eq:prefix_score}.  Algorithm~\ref{alg:wy} also states the equivalence check we run before
every experiment: the WY product must agree with the sequential product, and the
resulting map must preserve inner products, both to within \(10^{-4}\).

Reflection vectors are initialized in adjacent duplicate pairs,
\(\mathbf{v}_{2i-1}=\mathbf{v}_{2i}\).  Since \(H(\mathbf{v})H(\mathbf{v})=I\),
this makes \(R=I\) exactly at initialization, so training begins at the frozen
encoder and the adapter can only depart from it by descending the objective;
the paired vectors receive independent gradients and separate immediately.

The adapter holds \(FD\) free parameters, independent of the size of the
downstream index.  At \(F=64\) and \(D=2048\) this is \(131{,}072\) values
against the \(4{,}196{,}352\) of a single dense adaptor layer of the same width,
a factor of \(32\).  At deployment the WY factors are collapsed once into an
explicit matrix \(R\) and the recursion is never evaluated again: documents are
rotated offline, and a query resolved at width \(m\) costs \(O(Dm)\).

% \begin{algorithm}[t]
% \caption{Compact WY construction and application of an $F$-reflection
% orthogonal adapter ($F=64$ throughout).}
% \label{alg:wy}
% \begin{algorithmic}[1]
% \Require free parameters $U\in\mathbb{R}^{F\times D}$; batch $X\in\mathbb{R}^{n\times D}$
% \Ensure $XR^{\top}$ with $R^{\top}=H(\mathbf v_1)\cdots H(\mathbf v_F)$ orthogonal
% \Statex
% \Function{WYFactors}{$U$}
%   \State $V\gets\big(\text{row-normalize}(U)\big)^{\!\top}$
%          \Comment{$D\times F$, unit columns}
%   \State $\Gamma\gets V^{\top}V$ \Comment{$F\times F$ Gram, computed once}
%   \State $T\gets\mathbf{0}\in\mathbb{R}^{F\times F}$
%   \For{$j=1$ \textbf{to} $F$}
%     \State $T_{jj}\gets 2$
%     \If{$j>1$}
%       \State $T_{1:j-1,\,j}\gets-2\,T_{1:j-1,\,1:j-1}\,\Gamma_{1:j-1,\,j}$
%     \EndIf
%   \EndFor
%   \State \Return $V,T$
% \EndFunction
% \Statex
% \Function{Apply}{$X;V,T$}
%   \State \Return $X-\big((XV)\,T\big)V^{\top}$
% \EndFunction
% \Statex
% \Function{Materialize}{$V,T$} \Comment{once, at deployment}
%   \State \Return $\Call{Apply}{I_D;V,T}^{\top}$
% \EndFunction
% \Statex
% \Function{Verify}{$V,T$} \Comment{run before every experiment}
%   \State draw unit rows $Y\in\mathbb{R}^{64\times D}$;\;
%          $\tilde Y\gets\Call{Apply}{Y;V,T}$
%   \State \textbf{assert} $\lVert \tilde Y\tilde Y^{\top}-YY^{\top}\rVert_{\max}\le 10^{-4}$
%          \Comment{inner products preserved}
%   \State \textbf{assert} $\tilde Y$ matches the sequential product of
%          $H(\mathbf v_1),\dots,H(\mathbf v_F)$ to $10^{-4}$
% \EndFunction
% \end{algorithmic}
% \end{algorithm}

\begin{algorithm}[t]
\caption{Compact WY construction and application of the $F$-reflection
orthogonal adapter ($F=64$).}
\label{alg:wy}
\small
\begin{algorithmic}[1]
\Require $U\in\mathbb{R}^{F\times D}$; batch $X\in\mathbb{R}^{n\times D}$
\Ensure $XR^{\top}$, $R^{\top}=H(\mathbf v_1)\cdots H(\mathbf v_F)$
\Statex
\Function{WYFactors}{$U$}
  \State $V\gets\big(\text{row-normalize}(U)\big)^{\!\top}$ \Comment{unit columns}
  \State $\Gamma\gets V^{\top}V$ \Comment{Gram, once}
  \State $T\gets\mathbf{0}\in\mathbb{R}^{F\times F}$
  \For{$j=1$ \textbf{to} $F$}
    \State $T_{jj}\gets 2$
    \If{$j>1$}
      \State $T_{1:j-1,j}\gets-2T_{1:j-1,1:j-1}\Gamma_{1:j-1,j}$
    \EndIf
  \EndFor
  \State \Return $V,T$
\EndFunction
\Statex
\Function{Apply}{$X;V,T$}
  \State \Return $X-\big((XV)\,T\big)V^{\top}$
\EndFunction
\Statex
\Function{Materialize}{$V,T$} \Comment{at deployment}
  \State \Return $\Call{Apply}{I_D;V,T}^{\top}$
\EndFunction
\Statex
\Function{Verify}{$V,T$} \Comment{before each run}
  \State draw unit rows $Y\in\mathbb{R}^{64\times D}$
  \State $\tilde Y\gets\Call{Apply}{Y;V,T}$
  \State \textbf{assert} $\lVert\tilde Y\tilde Y^{\top}-YY^{\top}\rVert_{\max}\le10^{-4}$
  \State \textbf{assert} $\tilde Y$ matches $YH(\mathbf v_1)\cdots H(\mathbf v_F)$
\EndFunction
\end{algorithmic}
\end{algorithm}

\subsubsection{Two invariants}
Orthogonality gives the adapter two properties that an unconstrained map does
not have, and both are used by the training objective.

First, \emph{endpoint invariance}.  For an orthogonal \(R\) we have \(R^\top R=I\), so
\begin{equation}
    \langle Rx,Ry\rangle=x^{\top}R^{\top}Ry=\langle x,y\rangle ,
    \label{eq:orthogonal_endpoint}
\end{equation}
and therefore \(s^R_D(q,d)=s^I_D(q,d)\) for every pair.  A system that declines to truncate
recovers the frozen encoder exactly, so the adapter is safe to install ahead of
an index whose full-width behavior is already validated.  It also bounds what
the adapter can claim: any gain it reports is a gain under truncation, never a
change to the underlying retrieval quality.

Second, a \emph{conservation identity}.  Writing
\(z_j=(R\qvec)_j(R\dvec)_j\) for the per-coordinate contribution,
\begin{equation}
    \langle\qvec,\dvec\rangle
    =\underbrace{\sum_{j\le m} z_j}_{\text{prefix}}
    +\underbrace{\sum_{j>m} z_j}_{\text{tail}},
    \label{eq:conservation}
\end{equation}
where the left-hand side does not depend on \(R\).  The prefix residual is
therefore \emph{exactly} the tail sum: the total is fixed for each pair and only
its distribution over coordinates is learnable.  Because \(z_j\) is signed,
conservation binds each pair separately and not the margin between two pairs, so
the adapter can raise the prefix score of a relevant pair while lowering that of
a competitor.  This motivates the effective prefix information rate
\(\gamma_\tau\) analyzed in Section~\ref{sec:prefix_sufficiency_bound}: the
adapter cannot change a full-width score, only redistribute it across
coordinates.

\subsubsection{Training objective}
Let \(\mathcal{B}\) be a minibatch of training queries, each paired with
\(c=32\) candidates drawn from its shortlist and a multi-hot positive mask
\(P\).  Index the candidates of query \(b\) by \(j\), and let
\(\mathcal{G}_{bj}=\langle\qvec_b,\dvec_{bj}\rangle\) be the frozen full-width
scores and \(S^{(m)}_{bj}=s^R_m(q_b,d_{bj})\) the adapted prefix scores.  We
optimize a sum of three terms over a set of prefix widths
\(\mathcal{M}=\{64,128,256\}\).

The first is a multi-positive contrastive term applied at each width,
\begin{equation}
    \mathcal{L}_{\mathrm{rank}}
    =\frac{1}{|\mathcal{M}|}\sum_{m\in\mathcal{M}}
     \operatorname*{mean}_{b\in\mathcal{B}}
     \left[
       \log\!\sum_{j} e^{S^{(m)}_{bj}/\theta}
       -\log\!\!\sum_{j:P_{bj}=1}\!\! e^{S^{(m)}_{bj}/\theta}
     \right].
    \label{eq:rank_loss}
\end{equation}

The remaining two ask the truncated representation to reproduce the
full-width geometry.  Under an orthogonal map this target is unusually
well posed: by endpoint invariance, \(\mathcal{G}\) is simultaneously the frozen
encoder's score matrix and the adapter's own full-width output, so it is a
teacher that is exact that does not drift during training.

Two choices in how that residual is measured deserve to be stated, because
neither follows automatically from \eqref{eq:conservation}.  First, the
conserved quantity is the raw partial inner product \(\sum_{j\le m}z_j\),
whereas both terms below are measured on the renormalized prefix cosine of
\eqref{eq:prefix_score}.  A per-prefix norm is not part of the conserved total,
so driving \(S^{(m)}\) toward \(\mathcal{G}\) is a surrogate for the conservation
residual.  We measure it on cosines because
that is the quantity the deployed system ranks by; a training target on a
different scale from the scoring rule would optimize something the retriever
never computes.  Second, we weight each candidate by how contestable it is at
full width,
\begin{equation}
    W_{bj}=\exp\!\Big(-\big(\textstyle\max_{j'}\mathcal{G}_{bj'}-\mathcal{G}_{bj}\big)/\theta_b\Big),
    \label{eq:difficulty}
\end{equation}
held fixed (no gradient), so that the penalty concentrates on pairs a prefix
could plausibly reorder.  This gives
\begin{align}
    \mathcal{L}_{\mathrm{score}}
    &=\frac{1}{|\mathcal{M}|}\sum_{m}
      \operatorname*{mean}_{b}
      \frac{\sum_{j}W_{bj}\,\big|S^{(m)}_{bj}-\mathcal{G}_{bj}\big|}
           {\sum_{j}W_{bj}},
    \label{eq:score_loss}\\[2pt]
    \mathcal{L}_{\mathrm{geom}}
    &=\frac{1}{|\mathcal{M}|}\sum_{m}
      \Big\lVert
        \Pi_m\Pi_m^{\top}-QQ^{\top}
      \Big\rVert_{1},
    \label{eq:geom_loss}
\end{align}
where \(Q\) stacks the frozen queries of the batch, \(\Pi_m\) stacks
\(\pi_m(R\qvec_b)\), and \(\lVert\cdot\rVert_{1}\) denotes the entrywise mean
absolute deviation.  The objective is
\begin{equation}
    \mathcal{L}
    =\mathcal{L}_{\mathrm{rank}}
    +\lambda_{1}\mathcal{L}_{\mathrm{score}}
    +\lambda_{2}\mathcal{L}_{\mathrm{geom}}.
    \label{eq:adapter_objective}
\end{equation}
\(\lambda_1,\lambda_2>0\) weight the two preservation terms, \(\theta\) is
the temperature of the contrastive term and \(\theta_b\) that of the difficulty
weighting; Section~\ref{sec:eval} gives their values and the sweep they
come from. Algorithm~\ref{alg:adapter_train} gives the full procedure.  The same
optimizer, schedule, candidate pool, widths and contrastive term are used for
the unconstrained baseline, so that comparisons isolate the parameterization.

\begin{algorithm}[t]
\caption{Adapter training.  The encoder is frozen; only $U$ is learned.}
\label{alg:adapter_train}
\begin{algorithmic}[1]
\Require frozen train embeddings $Q\in\mathbb{R}^{n\times D}$, candidate index
         $J\in\mathbb{N}^{n\times c}$, positive mask $P$
\Require widths $\mathcal{M}$, weights $\lambda_1,\lambda_2$, temperatures
         $\theta,\theta_b$, epochs $n_{\mathrm{ep}}$, batch size $B$
\Ensure orthogonal $R$
\State $U\gets$ paired unit vectors so that $R=I$
       \Comment{$\mathbf v_{2i-1}=\mathbf v_{2i}$}
\For{$e=1$ \textbf{to} $n_{\mathrm{ep}}$}
  \ForAll{minibatches $\mathcal{B}$ of a random permutation of $[n]$}
    \State $Q_{\mathcal B}\gets Q[\mathcal B]$;\;
           $C_{\mathcal B}\gets\textsc{Corpus}[J[\mathcal B]]$
           \Comment{frozen}
    \State $\mathcal{G}\gets Q_{\mathcal B}C_{\mathcal B}^{\top}$
           \Comment{full-width teacher; fixed}
    \State $W\gets\exp\!\big(-(\operatorname{rowmax}\mathcal{G}-\mathcal{G})/\theta_b\big)$;\;
           detach $W$
    \State $V,T\gets\Call{WYFactors}{U}$
    \State $\tilde Q\gets\Call{Apply}{Q_{\mathcal B};V,T}$;\;
           $\tilde C\gets\Call{Apply}{C_{\mathcal B};V,T}$
    \For{$m\in\mathcal{M}$}
      \State $S^{(m)}\!\gets$ prefix cosines of $\tilde Q,\tilde C$ at width $m$
    \EndFor
    \State $\mathcal{L}\gets
            \mathcal{L}_{\mathrm{rank}}
            +\lambda_1\mathcal{L}_{\mathrm{score}}
            +\lambda_2\mathcal{L}_{\mathrm{geom}}$
            \Comment{Eqs.~\eqref{eq:rank_loss}--\eqref{eq:geom_loss}}
    \State AdamW step on $U$; clip $\lVert\nabla\rVert_2$ at $5$
  \EndFor
\EndFor
\State \Call{Verify}{$V,T$};\;
       \Return \Call{Materialize}{$V,T$}
\end{algorithmic}
\end{algorithm}

% ===========================================================================
\subsection{Per-Query Prefix Router}
\label{sec:router}

The adapter changes how quickly accuracy accrues with width; it does not choose
a width.  Queries vary in the width at which their rankings stabilize, so a
fixed prefix must be set for the hardest queries in the workload.  The router
recovers the difference.  Section~\ref{sec:prefix_sufficiency_bound} later
formalizes this variation and explains why per query allocation is useful.

\subsubsection{Escalation over an action pair}
Retrieval proceeds in two stages.  Every query is first ranked at a fixed stage-1
width \(m_0=64\) over its shortlist; this pass is unconditional and its cost is
charged to every query.  A router then predicts whether re-ranking the same
shortlist at a wider prefix would change the top-\(10\), and only escalated
queries pay for the second pass.  Concretely, an \emph{action pair}
\((\ell,h)\) with \(m_0\le\ell<h\) fixes the two widths a query may be resolved
at, and the policy escalates a fraction \(\rho\) of queries, giving
\begin{equation}
    \overline{m}(\rho)=(1-\rho)\,\ell+\rho\,h,
    \label{eq:mean_width}
\end{equation}
mean retained coordinates per query.  Because \(\ell\ge m_0\) and queries and
documents share one coordinate order, the stage-1 partial sums and partial norms
are reused, so a query resolved at width \(m\) costs \(m\) rather than
\(m_0+m\).  Sweeping \(\rho\in[0,1]\) traces a curve
whose two endpoints are exactly the fixed policies at \(\ell\) and at \(h\), so
a router is only credited for what it adds over static truncation at the same
average cost.  We use
\(\mathcal{A}=\{(64,128),(64,256),(128,256)\}\).

\subsubsection{Features}
The router sees only order statistics of the stage-1 score vector.  No relevance
label, no document identity, and no second encoder pass is involved: the
features are read off a ranking the system has already computed, so routing adds
no work that a single-stage system would not have done.  Let
\(\sigma_1\ge\sigma_2\ge\cdots\ge\sigma_{500}\) be the sorted stage-1 scores,
\(p=\sigma_{1:50}\), and \(\mathbb{H}=-\sum_i\eta_i\log\eta_i\) the entropy of
\(\eta=\mathrm{softmax}(p/\theta)\), at the same temperature as the contrastive
term.  The \(18\) features are
\begin{align*}
\text{levels:}\;\; &\sigma_1,\sigma_2,\sigma_5,\sigma_{10},\sigma_{11},\sigma_{20};\\
\text{gaps:}\;\; &\sigma_1-\sigma_2,\;\sigma_1-\sigma_{10},\;
                    \sigma_{10}-\sigma_{11},\;\sigma_{10}-\sigma_{15};\\
\text{moments:}\;\; &\operatorname{mean}(\sigma_{1:10}),\;\operatorname{std}(\sigma_{1:10}),\;
                    \operatorname{mean}(p),\;\operatorname{std}(p);\\
\text{shape:}\;\; &\mathbb{H},\;e^{\mathbb{H}},\;
                    \tfrac{\sigma_1-\operatorname{mean}(p)}{\operatorname{std}(p)},\;
                    \tfrac{\operatorname{std}(\sigma_{1:100})}
                          {\lvert\operatorname{mean}(\sigma_{1:100})\rvert}.
\end{align*}
The gap \(\sigma_{10}-\sigma_{11}\) straddles the evaluated cutoff and is the quantity a
wider prefix must overcome to change NDCG@10; \(e^{\mathbb{H}}\) is an effective count of
plausible candidates; and the scale-free last two features let one router
transfer across tasks and encoders whose score distributions differ in location
and spread.

\subsubsection{Fitting and use}
The supervision signal is the \emph{escalation gain}
\begin{equation}
    g(q)=\mathrm{NDCG@10}_h(q)-\mathrm{NDCG@10}_\ell(q),
    \label{eq:escalation_gain}
\end{equation}
computed on training queries only, where relevance labels are available.  A
gradient-boosted regression ensemble is fitted to predict \(g\) from the
features; at query time its prediction ranks the evaluation queries and the top
\(\rho\) fraction is escalated (Algorithm~\ref{alg:router}).  Because
\(g\) is a difference of two quantities the training split already measures,
fitting the router requires no additional forward passes through the encoder.

Two diagnostics accompany the reported curve.  Replacing the ensemble with a
multilayer perceptron on the same features, and with the router of a prior
operating-point study, tests whether the result depends on the predictor rather
than on the features; we report the upper envelope over the three, with seeds
averaged before the envelope is taken.  Replacing the prediction with the true
\(g\) yields an \emph{oracle} policy that upper-bounds any router restricted to
the same action pair, and separates the question of whether per-query allocation
helps at all from the question of whether the current predictor is what limits
it.

\begin{algorithm}[t]
\caption{Router fitting and query-time routed retrieval.}
\label{alg:router}
\begin{algorithmic}[1]
\Require adapter $R$; train split with labels; action pair $(\ell,h)$;
         escalation fraction $\rho$
\Statex\textbf{Fitting (once per action pair)}
\State transform and store the corpus at width $h$ using $R$
\ForAll{training queries $q$}
  \State $\phi(q)\gets\Call{Features}{\text{stage-1 ranking of }q\text{ at }m_0}$
  \State $g(q)\gets\mathrm{NDCG@10}_h(q)-\mathrm{NDCG@10}_\ell(q)$
\EndFor
\State fit $f$ on $\{(\phi(q),g(q))\}$
       \Comment{gradient-boosted trees}
\Statex\textbf{Retrieval (per query)}
\State $\qvec\gets R\,E(q)$
\State $\sigma\gets$ scores of $\pi_{m_0}(\qvec)$ against $\mathcal{C}(q)$,
       sorted
\State $\hat g\gets f(\Call{Features}{\sigma})$
\If{$\hat g$ is in the top $\rho$ fraction of the workload}
  \State re-rank $\mathcal{C}(q)$ at width $h$
\Else
  \State re-rank $\mathcal{C}(q)$ at width $\ell$
\EndIf
\State \Return top $10$
\end{algorithmic}
\end{algorithm}

\paragraph{Complexity.}
Rotating a query reads only the leading $h$ coordinates and costs $O(Dh)$, once
per query and independent of the shortlist. Because queries and documents share
one coordinate order, stage~1 leaves usable partial sums, so a query resolved at
width $m$ scores its shortlist in $O(|\mathcal{C}(q)|\,m)$ and the workload
averages $\Theta(|\mathcal{C}(q)|\,\overline{m}(\rho))$;
routing adds a partial sort of the stage-1 scores and one ensemble evaluation on
$18$ features. The skipped coordinates form a contiguous tail, so the index
stores $Nh$ values and a query reads exactly the values it multiplies, whereas a
scattered query-specific subset keeps $ND$ and gathers from a full-width record.
Offline, the corpus is rotated once in $O(NFD + NF^2)$, and the adapter holds
$FD$ parameters against the $\Theta(D^2)$ of a dense map.

\section{Prefix Sufficiency Analysis}
\label{sec:prefix_sufficiency_bound}
 
We now analyze how much width a query requires.  The analysis relates the
required width to corpus size, retrieval depth, and the rate at which
discriminative information accumulates in the prefix.  These results provide
a theoretical justification for the adapter and router in AdaWidth.
 
\subsection{Per-Query Bound}
We formalize the width required for a relevant document to enter the
retrieved set.  Let \(s_m(q,d)\) be the prefix score \eqref{eq:prefix_score}
between query \(q\) and document \(d\) at width \(m\).  For a relevant
document \(d^+\) and a non-relevant document \(d_i^-\), define the prefix margin
\begin{equation}
    \Delta_{qi}(m)
    = s_m(q,d^+) - s_m(q,d_i^-).
    \label{eq:prefix_margin}
\end{equation}
Here \(N\) is the number of documents in the corpus and \(k\) the retrieval
depth.  We write \(M_{\mathrm{first}}(q)\) for the smallest width at which
\(d^{+}\) enters the Top-\(k\) result and stays there, and \(\mathrm{q}_\tau[\cdot]\) for
the \(\tau\)-quantile of a quantity over the query workload; \(\tau\) is used
throughout this section for a quantile level and never for a temperature.
A non-relevant document overtakes \(d^+\) whenever
\(\Delta_{qi}(m)\leq 0\).  The number of such overtakers is
\begin{equation}
    Z_q(m)
    = \sum_{i=1}^{N-1}
      \mathbb{I}\!\left[\Delta_{qi}(m)\leq 0\right].
    \label{eq:overtaker_count}
\end{equation}
Thus, \(d^+\) lies outside the Top-\(k\) result at width \(m\) exactly when
\(Z_q(m)\geq k\).  This is an order-statistic characterization of prefix
sufficiency: the decision is controlled by the \(k\)-th strongest competing
document rather than by reconstruction error of an individual similarity
score.  All probabilities below are taken over the draw of the competing
documents from the corpus distribution, with \(q\) and \(d^{+}\) held fixed.
 
\begin{proposition}[Prefix sufficiency]
\label{prop:prefix_sufficiency}
Suppose that, over the prefix range relevant to Top-\(k\) retrieval, the
probability that any individual non-relevant document overtakes \(d^+\)
decreases exponentially with width:
\begin{equation}
    \Pr\!\left[\Delta_{qi}(m)\leq 0\right]
    \leq A_q \exp(-\gamma_q m),
    \label{eq:overtaking_decay}
\end{equation}
where \(A_q>0\) is a per-query constant and \(\gamma_q>0\) is an effective
\emph{prefix information rate}. Then
\begin{equation}
    \Pr[Z_q(m)\geq k]
    \leq \frac{N A_q \exp(-\gamma_q m)}{k},
    \label{eq:topk_failure_bound}
\end{equation}
and a sufficient condition for a failure probability of at most \(\delta\) is
\begin{equation}
    m
    \geq
    \frac{\log A_q-\log\delta}{\gamma_q}
    +\frac{1}{\gamma_q}\log N
    -\frac{1}{\gamma_q}\log k.
    \label{eq:simple_prefix_bound}
\end{equation}
\end{proposition}
 
\begin{proof}
\(Z_q(m)\) is a sum of \(N-1\) indicators, so expectation passes through it
term by term, and \eqref{eq:overtaking_decay} bounds each term:
\begin{equation}
    \mathbb{E}[Z_q(m)]
    =\sum_{i=1}^{N-1}\Pr\!\left[\Delta_{qi}(m)\leq 0\right]
    \leq N A_q \exp(-\gamma_q m).
    \label{eq:expected_overtakers}
\end{equation}
Linearity holds whatever the joint law of the competing documents, which
matters here because a retrieval corpus contains near-duplicates whose
overtaking events are strongly dependent.

The count \(Z_q(m)\) is non-negative, so Markov's inequality applies at the
threshold \(k\),
\[
    \Pr[Z_q(m)\geq k]\leq\frac{\mathbb{E}[Z_q(m)]}{k},
\]
which together with \eqref{eq:expected_overtakers} gives
\eqref{eq:topk_failure_bound}.  A Chernoff bound would replace the factor
\(1/k\) by an exponentially decaying one, but it requires the independence
that the same near-duplicates rule out.

It remains to solve for the width.  Requiring the right-hand side of
\eqref{eq:topk_failure_bound} to be at most \(\delta\) and taking logarithms,
\begin{equation}
    \frac{N A_q}{k}\exp(-\gamma_q m)\leq\delta
    \iff
    m\geq\frac{1}{\gamma_q}\log\frac{N A_q}{\delta k},
    \label{eq:solve_for_m}
\end{equation}
where \(\gamma_q>0\) makes the division preserve the direction of the
inequality.  Expanding the logarithm into its three terms gives
\eqref{eq:simple_prefix_bound}.
\end{proof}
 
Since \(M_{\mathrm{first}}>m\) requires failure at some width \(m'\ge m\),
summing \eqref{eq:topk_failure_bound} over \(m'\) gives
\begin{equation}
    \Pr[M_{\mathrm{first}}>m]
    \leq
    \frac{N A_q}{k}\cdot\frac{\exp(-\gamma_q m)}{1-\exp(-\gamma_q)},
    \label{eq:first_hit_tail}
\end{equation}
which changes only the constant term of \eqref{eq:simple_prefix_bound} and
leaves the coefficients of \(\log N\) and \(\log k\) unchanged.  Applied to a
\(\tau\)-level query-difficulty envelope with \(\delta=1-\tau\), it predicts
that the required width increases logarithmically with corpus size and
decreases logarithmically with retrieval depth.
 
\subsection{Scaling over the Workload}
The basic bound assumes that the nominal corpus size and Top-\(k\) tolerance
translate directly into independent competition opportunities.  Correlated
documents, local density, and duplicated content can instead produce an
effective failure envelope
\begin{equation}
    \Pr[Z_q(m)\geq k]
    \leq
    \kappa_\tau N^{\alpha_\tau}k^{-\beta_\tau}
    \exp(-\gamma_\tau m).
    \label{eq:generalized_failure_envelope}
\end{equation}
Here \(\kappa_\tau>0\) is a constant, which also absorbs the geometric factor of
\eqref{eq:first_hit_tail}, and \(\gamma_\tau>0\) is the information rate of the
\(\tau\)-level envelope, the workload-level counterpart of the per-query
\(\gamma_q\) in \eqref{eq:overtaking_decay}, while \(\alpha_\tau,\beta_\tau>0\)
respectively describe the growth of
effective competition and the tolerance provided by a deeper retrieved set.
Solving this bound for \(m\) yields the three-parameter response
\begin{equation}
    \mathrm{q}_\tau[M_{\mathrm{first}}(N,k)]
    \lesssim
    a_\tau+b_\tau\log N+c_\tau\log k,
    \label{eq:three_parameter_scaling}
\end{equation}
with
\begin{equation}
    a_\tau=\frac{\log(\kappa_\tau/\delta)}{\gamma_\tau},
    \qquad
    b_\tau=\frac{\alpha_\tau}{\gamma_\tau},
    \qquad
    c_\tau=-\frac{\beta_\tau}{\gamma_\tau}.
    \label{eq:scaling_coefficients}
\end{equation}
The proposition therefore predicts the signs \(b_\tau>0\) and \(c_\tau<0\), while
allowing their magnitudes to vary across datasets and encoders.  It also
assigns a geometric meaning to the corpus-size slope: representations that
accumulate discriminative information more rapidly in their leading
coordinates have larger \(\gamma_\tau\) and hence smaller \(b_\tau\).
 
This interpretation explains the role of our orthogonal adapter.  By
\eqref{eq:orthogonal_endpoint}, the adapter cannot improve retrieval by changing the full-width endpoint;
instead, it reorganizes existing discriminative signal toward earlier
coordinates.  In the bound above, this corresponds to increasing the
effective prefix information rate and flattening the growth of the required
width with corpus size.

\section{Evaluation}
\label{sec:eval}

Our evaluation answers four questions. Section~\ref{sec:eval_scaling_validation} tests whether the scaling form predicted by the prefix-sufficiency bound in Section~\ref{sec:prefix_sufficiency_bound} holds empirically.
Sections~\ref{sec:eval_text} and~\ref{sec:eval_vqa} compare the accuracy retained by AdaWidth at different widths with three baselines on text and vision--language retrieval. Section~\ref{sec:eval_capacity} examines the effect of adapter capacity, followed by the ablation study in
Section~\ref{sec:eval_ablation}.

\textit{Setup.}
All experiments run on a single NVIDIA A100-SXM4 GPU with 40\,GB of device
memory, hosted on a node with 32 vCPUs and 115\,GB of system memory.  Encoding, adapter training and evaluation use PyTorch 2.5.1 with CUDA 12.1.

\textit{Datasets.}
We evaluate on four text retrieval tasks from BEIR~\cite{thakur2021beir} and two vision--language retrieval tasks, summarized in
Table~\ref{tab:datasets}.

\begin{table}[t]
    \centering
    \caption{Evaluation datasets.}
    \label{tab:datasets}
    \small
    \renewcommand{\arraystretch}{1.08}
    \begin{tabular*}{\columnwidth}{@{\extracolsep{\fill}}llr@{}}
        \toprule
        Dataset & Retrieval task & Corpus size \\
        \midrule
        FiQA~\cite{maia201818}
            & Financial question answering
            & 57{,}638 \\
        ArguAna~\cite{wachsmuth2018retrieval}
            & Counterargument retrieval
            & 8{,}674 \\
        Quora~\cite{thakur2021beir}
            & Duplicate-question retrieval
            & 522{,}931 \\
        MS~MARCO~\cite{nguyen2016ms}
            & Web passage retrieval
            & 2{,}000{,}000$^{\dagger}$ \\
        A-OKVQA~\cite{schwenk2022okvqa}
            & Knowledge retrieval for VQA
            & 30{,}803 \\
        OK-VQA~\cite{marino2019ok}
            & Knowledge retrieval for VQA
            & 15{,}039 \\
        \bottomrule
        \multicolumn{3}{@{}l}{\footnotesize
        $^{\dagger}$ A 2M-document subsample.}
    \end{tabular*}
\end{table}

\textit{Encoders.}
We evaluate four frozen text encoders: E5-Mistral-7B~\cite{e5Mistral} with 7.1B parameters and 4096 dimensions, Qwen3-Embedding-0.6B~\cite{qwen3technicalreport} with 0.6B parameters and 1024 dimensions, Nomic-embed-v1.5~\cite{nussbaum2024nomic} with 137M parameters and 768 dimensions, and BGE-large-en-v1.5~\cite{bge_embedding} with 335M parameters and 1024 dimensions. We also evaluate the frozen vision language encoder Qwen3-VL-Embedding-2B~\cite{qwen3technicalreport} with 2B parameters and 2048 dimensions. Qwen3-Embedding-0.6B, Nomic-embed-v1.5, and Qwen3-VL-Embedding-2B are trained with Matryoshka representation learning, whereas the other encoders are not. AdaWidth and all baselines are trained independently on each encoder.

\textit{Baselines.}
We compare against three learned methods, Matryoshka Adaptor~\cite{yoon2024matryoshka}, SMEC~\cite{zhang2025smec} and Learning-to-Select~\cite{wu2026learning}, each retrained on every encoder and dataset. We also report prefix truncation, which keeps the first $m$ coordinates of the frozen embedding.

\textit{Hyperparameters.}
The adapter's loss weights and temperatures were swept on SciFact~\cite{scifact} and NFCorpus~\cite{nfcorpus}, two BEIR tasks we measure but do not report. We considered
$\lambda_1,\lambda_2\in\{1,\allowbreak 2,\allowbreak 4,\allowbreak 8\}$,
$\theta_b\in\{0.02,\allowbreak 0.05,\allowbreak 0.10,\allowbreak 0.20\}$, and
$\theta\in\{0.02,\allowbreak 0.05,\allowbreak 0.10\}$, using three seeds for each setting. We fix $\lambda_1=2$, $\lambda_2=8$, and $\theta=\theta_b=0.05$, and use this single setting for every encoder, dataset, and modality reported below.

% \begin{table*}[t]
%     \centering
%     \caption{Validation of the predicted q90 scaling form.  For each of the
%     four text datasets, we report the encoder with the mean-point \(R^2\) among BGE-large, Qwen3-0.6B, E5-Mistral-7B and Nomic. The two vision-language datasets use Qwen3-VL-Embedding-2B.  Coefficients are dataset- and encoder-specific.}
%     \label{tab:q90_mean_point_validation}
%     \begin{tabular}{@{}lll c@{}}
%         \toprule
%         Dataset & Encoder & Fitted q90 response & Mean-point \(R^2\) \\
%         \midrule
%         FiQA
%             & BGE-large
%             & \(0.69+27.33\log N-40.77\log k\)
%             & \(0.967\) \\
%         ArguAna
%             & Qwen3-0.6B
%             & \(-59.41+33.28\log N-43.01\log k\)
%             & \(0.973\) \\
%         Quora
%             & E5-Mistral-7B
%             & \(-172.59+48.72\log N-52.23\log k\)
%             & \(0.974\) \\
%         MS MARCO
%             & Nomic
%             & \(-3.28+5.63\log N-6.81\log k\)
%             & \(0.926\) \\
%         OK-VQA
%             & Qwen3-VL-2B
%             & \(-12.42+18.98\log N-39.91\log k\)
%             & \(0.917\) \\
%         A-OKVQA
%             & Qwen3-VL-2B
%             & \(-23.97+17.55\log N-34.91\log k\)
%             & \(0.905\) \\
%         \bottomrule
%     \end{tabular}
% \end{table*}

\begin{figure*}[t]
    \centering
    \includegraphics[width=\textwidth]{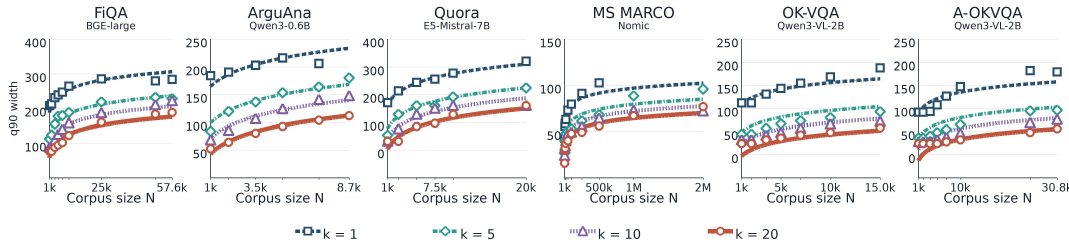}
    \caption{Validation of the predicted q90 scaling.}
    
    \label{fig:q90fit}
\end{figure*}

\subsection{Validation of the Predicted Width Scaling}
\label{sec:eval_scaling_validation}

We first test the functional form predicted by
Section~\ref{sec:prefix_sufficiency_bound}, independently of the adapter and
router.  For each dataset--encoder pair, we vary corpus size \(N\) and
retrieval depth \(k\), construct three positive-preserving corpus subsamples,
and measure the 90th percentile of \(M_{\mathrm{first}}\).  We then fit
\begin{equation}
    \mathrm{q}_{0.9}[M_{\mathrm{first}}(N,k)]
    =a+b\log N+c\log k
    \label{eq:q90_empirical_fit}
\end{equation}
by ordinary least squares.

Our primary statistic is \emph{mean-point} \(R^2\).  Specifically, we first
average the three observed q90 values at every \((N,k)\) setting and then
evaluate the fitted response surface:
\begin{equation}
    R^2_{\mathrm{mean}}
    =
    1-
    \frac{
        \sum_{N,k}
        \left(\bar m_{N,k}-\widehat m_{N,k}\right)^2
    }{
        \sum_{N,k}
        \left(\bar m_{N,k}-\bar m\right)^2
    }.
    \label{eq:mean_point_r2}
\end{equation}
As \(R^2\) is commonly interpreted as the fraction of response variation
explained by a fitted model~\cite{draper1998applied,montgomery2022introduction,chicco2021coefficient},
we treat \(R^2_{\mathrm{mean}}>0.90\) as a good fit in this analysis.  Under
Equation~\eqref{eq:mean_point_r2}, this criterion means that the fitted surface
explains more than \(90\%\) of the variation in the mean response across
operating conditions.  Averaging first removes variation due solely to the
sampled negative subset.  We complement this descriptive statistic with
grouped held-out tests in subsequent analyses.

Figure~\ref{fig:q90fit} shows a consistent directional
pattern across text and vision-language retrieval.  The four text tasks
obtain mean-point \(R^2\) values between \(0.926\) and \(0.973\), with FiQA, ArguAna, and Quora exceeding \(0.96\).  The same three-parameter form remains strong on OK-VQA (\(0.917\)) and captures the majority of the mean variation on A-OKVQA (\(0.905\)).  These results support a shared first-order scaling structure: the response to \(N\) and \(k\) is stable in form, while its magnitude remains representation- and dataset-dependent.

\subsection{Quality and Dimensionality}
\label{sec:eval_text}

\begin{figure*}[t]
    \centering
    \includegraphics[width=\textwidth]{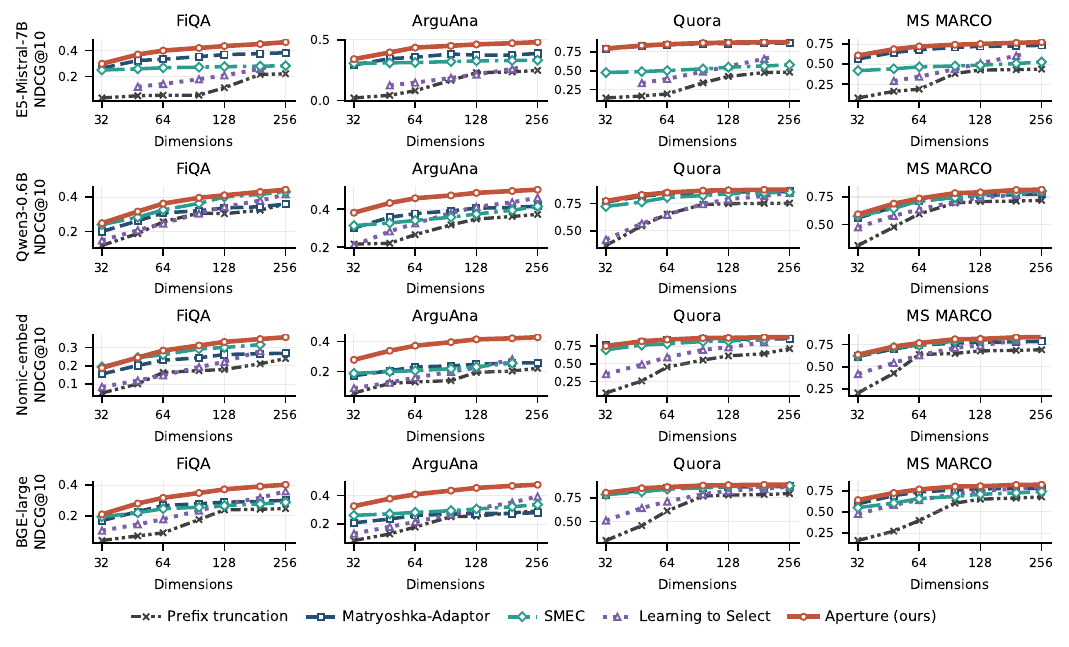}
    \caption{Cost against quality on the four tasks, one row per frozen encoder.  The x-axis is mean retained coordinates per query; the y-axis is NDCG@10.  Every curve is the mean of three seeds.}
    
    \label{fig:text-results}
\end{figure*}

Figure~\ref{fig:text-results} reports the four text tasks on four encoders. Averaged over
operating points the margin is $+5.88$ NDCG@10 against Matryoshka-Adaptor,
$+9.94$ against SMEC and $+15.77$ against Learning-to-Select.  Read the other
way, at $32$ retained coordinates Matryoshka-Adaptor needs $2.36\times$ the width
to reach the NDCG@10 AdaWidth already has, and at $48$ coordinates it needs
$3.04\times$; against Learning-to-Select the same two anchors give $4.96\times$
and $6.23\times$.  Ratios are computed for each task and encoder separately and aggregated by geometric mean. We quote these two anchors because they are the widths at which nothing is censored.  At $64$ coordinates and above a baseline sometimes never reaches AdaWidth's accuracy at any width it was measured at, and averaging only over the cases that do would silently drop the ones where the margin is largest.

The margin holds across both axes of the grid.  Broken down by task, the gain
over Matryoshka-Adaptor is $+11.70$ NDCG@10 on ArguAna, $+6.56$ on FiQA,
$+3.60$ on MS~MARCO and $+1.36$ on Quora; against SMEC the four tasks give
$+12.23$, $+6.41$, $+9.68$ and $+11.54$, and against Learning-to-Select
$+16.89$, $+12.25$, $+13.01$ and $+21.33$.  Broken down by encoder, the gain
over Matryoshka-Adaptor ranges from $+4.62$ on E5-Mistral-7B to $+7.24$ on
BGE-large, with Qwen3-0.6B at $+5.06$ and Nomic-embed at $+6.61$; the
corresponding figures are $+20.58$, $+8.17$, $+4.80$ and $+5.86$ against SMEC,
and $+26.82$, $+11.80$, $+10.65$ and $+15.17$ against Learning-to-Select.  The
three encoders trained with Matryoshka representation learning and the two
trained without it are all represented in this range, so the gain does not
depend on the frozen encoder already carrying a prefix structure.

In absolute terms, averaged over every task and encoder at which all methods
reach every anchor, AdaWidth attains NDCG@10 of $0.431$ at $32$ retained coordinates,
$0.499$ at $48$, $0.534$ at $64$, $0.576$ at $128$ and $0.597$ at $256$.  The
corresponding values are $0.377$, $0.438$, $0.472$, $0.500$ and $0.510$ for
Matryoshka-Adaptor, $0.393$, $0.433$, $0.462$, $0.503$ and $0.534$ for SMEC,
and $0.265$, $0.342$, $0.390$, $0.485$ and $0.550$ for Learning-to-Select.
AdaWidth at $64$ coordinates exceeds what Matryoshka-Adaptor reaches at $256$,
and at $96$ coordinates it exceeds what all three baselines reach at $256$.
Prefix truncation reaches $0.155$, $0.247$, $0.333$, $0.434$ and $0.462$ at the
same five widths, below every learned method at every width.

\subsection{Quality and Dimensionality in Vision Language Retrieval}
\label{sec:eval_vqa}

\begin{figure}[t]
    \centering
    \includegraphics[width=\columnwidth]{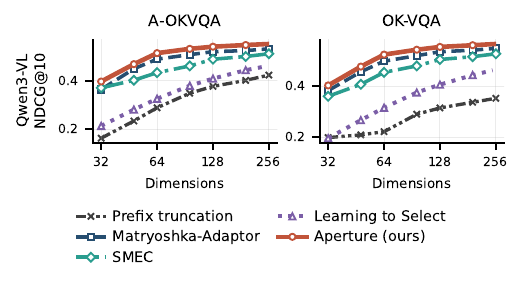}
    \caption{Cost against quality on the two knowledge-VQA retrieval tasks,
    with a frozen Qwen3-VL-Embedding-2B encoder.  Axes and protocol are those of Figure~\ref{fig:text-results}.}
    \label{fig:vqa-results}
\end{figure}

Figure~\ref{fig:vqa-results} tests that on two knowledge-VQA
retrieval tasks, where a question and an image are encoded jointly by a frozen
Qwen3-VL-2B and matched against a corpus of textual knowledge snippets.

AdaWidth is ahead at all $12$ shared operating points against all three
baselines, by $+2.27$ NDCG@10 on average against Matryoshka-Adaptor, $+5.68$
against SMEC and $+16.70$ against Learning-to-Select.

The two tasks behave alike.  On A-OKVQA, AdaWidth reaches NDCG@10 of $0.397$ at
$32$ retained coordinates, $0.514$ at $64$ and $0.540$ at $128$, against
$0.364$, $0.489$ and $0.519$ for Matryoshka-Adaptor, $0.370$, $0.433$ and
$0.488$ for SMEC, and $0.213$, $0.325$ and $0.409$ for Learning-to-Select.  On
OK-VQA the same three widths give $0.403$, $0.524$ and $0.554$ for AdaWidth,
against $0.383$, $0.498$ and $0.535$ for Matryoshka-Adaptor, $0.361$, $0.455$
and $0.505$ for SMEC, and $0.196$, $0.316$ and $0.406$ for Learning-to-Select.
Prefix truncation gives $0.163$, $0.288$ and $0.377$ on A-OKVQA and $0.199$,
$0.222$ and $0.315$ on OK-VQA.
The adapter is therefore worth $23$ and $30$ NDCG@10 at $64$ coordinates;
on both tasks AdaWidth at $64$ coordinates is above prefix truncation at
$256$.

The scaling form of Section~\ref{sec:eval_scaling_validation} also holds on both
tasks with the predicted signs, though with the weakest fits in the study
($R^2$ of $0.917$ and $0.905$ against $0.926$--$0.973$ on text).  Vision--language
retrieval is where we would expect the bound to be loosest: the competing
documents are far from independent, since a knowledge corpus assembled around a
question set contains many near-duplicates, and it is exactly that dependence
that the generalized envelope in
Equation~\eqref{eq:generalized_failure_envelope} absorbs into
$\alpha_\tau$ and $\beta_\tau$.

\subsection{Sensitivity to Adapter Capacity}
\label{sec:eval_capacity}

\begin{table}[t]
    \centering
    \caption{Quality against adapter size.  $F$ is the number of Householder reflections. Parameter counts are averaged over the same encoders.}
    \label{tab:capacity}
    \small
    \setlength{\tabcolsep}{6pt}
    \begin{tabular}{@{}lrr@{}}
        \toprule
        Configuration & Parameters & NDCG@10 \\
        \midrule
        Learning-to-Select     & $689{,}728$   & $.4604$ \\
        Matryoshka-Adaptor     & $2{,}901{,}312$ & $.4983$ \\
        SMEC                   & $3{,}870{,}379$ & $.4966$ \\
        \midrule
        AdaWidth, $F=8$        & $13{,}824$    & $.5716$ \\
        AdaWidth, $F=16$       & $27{,}648$    & $.5753$ \\
        AdaWidth, $F=32$       & $55{,}296$    & $.5787$ \\
        AdaWidth, $F=64$       & $110{,}592$   & $.5769$ \\
        AdaWidth, $F=128$      & $221{,}184$   & $.5684$ \\
        \bottomrule
    \end{tabular}
\end{table}

% Table~\ref{tab:capacity} varies the one setting that controls how much the
% adapter can learn.  Sweeping $F$ from $8$ to $128$ moves the parameter count
% over a sixteen-fold range and moves NDCG@10 by $1.03$; every setting in that
% range is above every baseline, and the weakest of them leads the strongest
% baseline by $7.01$ NDCG@10.  The result is uniform across the grid: at each of the five capacities, AdaWidth is above all three
% baselines on every task and encoder, spanning corpus sizes from
% $8.7$K to $2$M documents.  Between $F=16$ and $F=64$ the spread is $0.34$
% NDCG@10, so the setting can be chosen from the deployment's parameter budget.

% The comparison also runs the other way.  At $F=8$ the adapter holds $13{,}824$ parameters, $50\times$ fewer than Learning-to-Select and $210\times$ and $280\times$ fewer than Matryoshka-Adaptor and SMEC, and still leads all three by more than $7$
% NDCG@10.  A product of $F$ reflections costs $F\!\cdot\!D$ parameters against the $\Theta(D^{2})$ a dense map requires, so the gap widens with the encoder
% width rather than closing: at $F=64$ the adapter is $12\times$ smaller than
% a dense $D\times D$ map on the $768$-dimensional encoder, $16\times$ on the
% $1024$-dimensional encoders and $64\times$ on the $4096$-dimensional one.  The
% advantage therefore grows as encoders do, and it is realized without the
% retraining a wider dense map would require.

Table~\ref{tab:capacity} shows that AdaWidth is insensitive to adapter capacity. Sweeping $F$ from $8$ to $128$ changes NDCG@10 by only $1.03$ points, and every setting remains above all baselines, with the weakest leading the strongest baseline by $7.01$ points. This pattern holds across all tasks and encoders, spanning corpus sizes from $8.7$K to $2$M documents. Between $F=16$ and $F=64$, the spread is only $0.34$ NDCG@10, allowing $F$ to be selected according to the available parameter budget.

At $F=8$, the adapter uses $13{,}824$ parameters, $50\times$, $210\times$, and $280\times$ fewer than Learning-to-Select, Matryoshka-Adaptor, and SMEC, respectively, while still leading them by more than $7$ NDCG@10 points. Since $F$ reflections require $FD$ parameters, compared with $\Theta(D^2)$ for a dense map, this advantage grows with encoder width. At $F=64$, AdaWidth is $12\times$, $16\times$, and $64\times$ smaller than a dense map for encoders of width $768$, $1024$, and $4096$, respectively.

\subsection{Ablation Studies}
\label{sec:eval_ablation}

\begin{table}[t]
    \centering
    \caption{Ablation. NDCG@10 averaged over all encoders and tasks.}
    \label{tab:ablation}
    \small
    \setlength{\tabcolsep}{4.5pt}
    \begin{tabular}{@{}lrrrrrr@{}}
        \toprule
        & \multicolumn{5}{c}{Retained coordinates} & \\
        \cmidrule(lr){2-6}
        & $64$ & $96$ & $128$ & $192$ & $256$ & Mean \\
        \midrule
        Prefix truncation      & $.3537$ & $.3870$ & $.4214$ & $.4225$ & $.4290$ & $.4027$ \\
        Matryoshka-Adaptor     & $.5327$ & $.5499$ & $.5603$ & $.5642$ & $.5837$ & $.5582$ \\
        SMEC                   & $.4802$ & $.5016$ & $.5181$ & $.5316$ & $.5343$ & $.5132$ \\
        Learning-to-Select     & $.3778$ & $.4439$ & $.4822$ & $.5292$ & $.5321$ & $.4730$ \\
        \midrule
        Router only            & $.3590$ & $.3926$ & $.4351$ & $.4433$ & $.4490$ & $.4158$ \\
        Adapter only           & $.5811$ & $.6036$ & $.6165$ & $.6278$ & $.6315$ & $.6121$ \\
        AdaWidth (R+A)         & $\mathbf{.5812}$ & $\mathbf{.6066}$ & $\mathbf{.6174}$ & $\mathbf{.6294}$ & $\mathbf{.6317}$ & $\mathbf{.6134}$ \\
        \bottomrule
        \multicolumn{7}{@{}l@{}}{\footnotesize R+A: Router $+$ Adapter.}
    \end{tabular}
\end{table}

Table~\ref{tab:ablation} shows that the adapter provides most of the gain, raising mean NDCG@10 from $0.4027$ under prefix truncation to $0.6121$ and exceeding the strongest baseline by $5.39$ points. The router alone improves prefix truncation by $1.31$ points, while its gain on top of the adapter is modest overall but reaches $0.45$ points on corpora with more than $10^5$ documents.

\section{Conclusion}
We presented AdaWidth, which allocates embedding width per query within a shared prefix representation.  Its orthogonal adapter moves discriminative signal toward leading coordinates, while its router evaluates wider prefixes only when the ranking is likely to change.  Our prefix sufficiency analysis relates the required width to corpus size and retrieval depth.  Across six retrieval tasks and five frozen encoders, AdaWidth matches state-of-the-art NDCG@10 using $55\%$--$84\%$ fewer coordinates per query.

% \begin{acks}
% \end{acks}

\bibliographystyle{ACM-Reference-Format}
\bibliography{Reference}

@String{Computing = "Computing" }

@String{Computer = "{IEEE} Computer" }

@String{Springer = "Springer-Verlag" }

@inproceedings{karpukhin2020dense,
    title = "Dense Passage Retrieval for Open-Domain Question Answering",
    author = "Karpukhin, Vladimir  and
      Oguz, Barlas  and
      Min, Sewon  and
      Lewis, Patrick  and
      Wu, Ledell  and
      Edunov, Sergey  and
      Chen, Danqi  and
      Yih, Wen-tau",
    editor = "Webber, Bonnie  and
      Cohn, Trevor  and
      He, Yulan  and
      Liu, Yang",
    booktitle = "Proceedings of the 2020 Conference on Empirical Methods in Natural Language Processing (EMNLP)",
    month = nov,
    year = "2020",
    address = "Online",
    publisher = "Association for Computational Linguistics",
    url = "https://aclanthology.org/2020.emnlp-main.550/",
    doi = "10.18653/v1/2020.emnlp-main.550",
    pages = "6769--6781"
}

@article{dasgupta2003elementary,
  title={An elementary proof of a theorem of Johnson and Lindenstrauss},
  author={Dasgupta, Sanjoy and Gupta, Anupam},
  journal={Random Structures \& Algorithms},
  volume={22},
  number={1},
  pages={60--65},
  year={2003},
  publisher={Wiley Online Library}
}

@article{li2026spectral,
  title={Spectral Tempering for Embedding Compression in Dense Passage Retrieval},
  author={Li, Yongkang and Eustratiadis, Panagiotis and Kanoulas, Evangelos},
  journal={arXiv preprint arXiv:2603.19339},
  year={2026}
}

@article{kusupati2022matryoshka,
  title={Matryoshka representation learning},
  author={Kusupati, Aditya and Bhatt, Gantavya and Rege, Aniket and Wallingford, Matthew and Sinha, Aditya and Ramanujan, Vivek and Howard-Snyder, William and Chen, Kaifeng and Kakade, Sham and Jain, Prateek and others},
  journal={Advances in Neural Information Processing Systems},
  volume={35},
  pages={30233--30249},
  year={2022}
}

@inproceedings{yoon2024search,
  title={Search-adaptor: Embedding customization for information retrieval},
  author={Yoon, Jinsung and Chen, Yanfei and Arik, Sercan and Pfister, Tomas},
  booktitle={Proceedings of the 62nd Annual Meeting of the Association for Computational Linguistics (Volume 1: Long Papers)},
  pages={12230--12247},
  year={2024}
}

@inproceedings{yoon2024matryoshka,
  title={Matryoshka-adaptor: Unsupervised and supervised tuning for smaller embedding dimensions},
  author={Yoon, Jinsung and Sinha, Rajarishi and Arik, Sercan O and Pfister, Tomas},
  booktitle={Proceedings of the 2024 Conference on Empirical Methods in Natural Language Processing},
  pages={10318--10336},
  year={2024}
}

@inproceedings{zhang2025smec,
  title={SMEC: rethinking matryoshka representation learning for retrieval embedding compression},
  author={Zhang, Biao and Chen, Lixin and Liu, Tong and Zheng, Bo},
  booktitle={Proceedings of the 2025 Conference on Empirical Methods in Natural Language Processing},
  pages={26220--26233},
  year={2025}
}

@article{zhao2026dive,
  title={DIVE: Embedding Compression via Self-Limiting Gradient Updates},
  author={Zhao, Dongfang},
  journal={arXiv preprint arXiv:2605.20689},
  year={2026}
}

@article{rege2023adanns,
  title={Adanns: A framework for adaptive semantic search},
  author={Rege, Aniket and Kusupati, Aditya and Fan, Alan and Cao, Qingqing and Kakade, Sham and Jain, Prateek and Farhadi, Ali and others},
  journal={Advances in Neural Information Processing Systems},
  volume={36},
  pages={76311--76335},
  year={2023}
}

@inproceedings{faggioli2024dime,
  title={Dimension importance estimation for dense information retrieval},
  author={Faggioli, Guglielmo and Ferro, Nicola and Perego, Raffaele and Tonellotto, Nicola},
  booktitle={Proceedings of the 47th International ACM SIGIR Conference on Research and Development in Information Retrieval},
  pages={1318--1328},
  year={2024}
}

@inproceedings{derasmo2024eclipse,
  title={Eclipse: Contrastive dimension importance estimation with pseudo-irrelevance feedback for dense retrieval},
  author={D'Erasmo, Giulio and Trappolini, Giovanni and Silvestri, Fabrizio and Tonellotto, Nicola},
  booktitle={Proceedings of the 2025 International ACM SIGIR Conference on Innovative Concepts and Theories in Information Retrieval (ICTIR)},
  pages={147--154},
  year={2025}
}

@inproceedings{faggioli2025codime,
  title={Codime: a counterfactual approach for dimension importance estimation through click logs},
  author={Faggioli, Guglielmo and Ferro, Nicola and Perego, Raffaele and Tonellotto, Nicola},
  booktitle={Proceedings of the 48th International ACM SIGIR Conference on Research and Development in Information Retrieval},
  pages={2330--2340},
  year={2025}
}

@inproceedings{derasmo2026rdime,
  title={Statistical Foundations of DIME: Risk Estimation for Practical Index Selection},
  author={D'Erasmo, Giulio and Campagnano, Cesare and Mallia, Antonio and Brutti, Pierpaolo and Tonellotto, Nicola and Silvestri, Fabrizio},
  booktitle={Proceedings of the 19th Conference of the European Chapter of the Association for Computational Linguistics (Volume 2: Short Papers)},
  pages={722--730},
  year={2026}
}

@inproceedings{wu2026learning,
  title={Learning to Select: Query-Aware Adaptive Dimension Selection for Dense Retrieval},
  author={Wu, Zhanyu and Zhang, Richong and Nie, Zhijie},
  booktitle={Proceedings of the 64th Annual Meeting of the Association for Computational Linguistics (Volume 1: Long Papers)},
  pages={18666--18677},
  year={2026}
}

@inproceedings{
    thakur2021beir,
    title={{BEIR}: A Heterogeneous Benchmark for Zero-shot Evaluation of Information Retrieval Models},
    author={Nandan Thakur and Nils Reimers and Andreas R{\"u}ckl{\'e} and Abhishek Srivastava and Iryna Gurevych},
    booktitle={Thirty-fifth Conference on Neural Information Processing Systems Datasets and Benchmarks Track (Round 2)},
    year={2021},
    url={https://openreview.net/forum?id=wCu6T5xFjeJ}
}

@inproceedings{e5Mistral,
  title={Improving text embeddings with large language models},
  author={Wang, Liang and Yang, Nan and Huang, Xiaolong and Yang, Linjun and Majumder, Rangan and Wei, Furu},
  booktitle={Proceedings of the 62nd annual meeting of the association for computational linguistics (volume 1: long papers)},
  pages={11897--11916},
  year={2024}
}

@misc{qwen3technicalreport,
      title={Qwen3 Technical Report}, 
      author={Qwen Team},
      year={2025},
      eprint={2505.09388},
      archivePrefix={arXiv},
      primaryClass={cs.CL},
      url={https://arxiv.org/abs/2505.09388}, 
}

@misc{nussbaum2024nomic,
      title={Nomic Embed: Training a Reproducible Long Context Text Embedder}, 
      author={Zach Nussbaum and John X. Morris and Brandon Duderstadt and Andriy Mulyar},
      year={2024},
      eprint={2402.01613},
      archivePrefix={arXiv},
      primaryClass={cs.CL}
}

@misc{bge_embedding,
      title={C-Pack: Packaged Resources To Advance General Chinese Embedding}, 
      author={Shitao Xiao and Zheng Liu and Peitian Zhang and Niklas Muennighoff},
      year={2023},
      eprint={2309.07597},
      archivePrefix={arXiv},
      primaryClass={cs.CL}
}

@article{maia201818,
  title={Www'18 open challenge: financial opinion mining and question answering},
  author={Maia, Macedo and Handschuh, Siegfried and Freitas, Andr{\'e} and Davis, Brian and McDermott, Ross and Zarrouk, Manel and Balahur, Alexandra},
  year={2018},
  publisher={Association for Computing Machinery}
}

@inproceedings{wachsmuth2018retrieval,
  title={Retrieval of the best counterargument without prior topic knowledge},
  author={Wachsmuth, Henning and Syed, Shahbaz and Stein, Benno},
  booktitle={Proceedings of the 56th Annual Meeting of the Association for Computational Linguistics (Volume 1: Long Papers)},
  pages={241--251},
  year={2018}
}

@article{nguyen2016ms,
  title={Ms marco: A human-generated machine reading comprehension dataset},
  author={Nguyen, Tri and Rosenberg, Mir and Song, Xia and Gao, Jianfeng and Tiwary, Saurabh and Majumder, Rangan and Deng, Li},
  year={2016}
}

@inproceedings{marino2019ok,
  title={Ok-vqa: A visual question answering benchmark requiring external knowledge},
  author={Marino, Kenneth and Rastegari, Mohammad and Farhadi, Ali and Mottaghi, Roozbeh},
  booktitle={2019 IEEE/CVF conference on computer vision and pattern recognition (CVPR)},
  pages={3190--3199},
  year={2019},
  organization={IEEE}
}

@inproceedings{schwenk2022okvqa,
  title={A-okvqa: A benchmark for visual question answering using world knowledge},
  author={Schwenk, Dustin and Khandelwal, Apoorv and Clark, Christopher and Marino, Kenneth and Mottaghi, Roozbeh},
  booktitle={European conference on computer vision},
  pages={146--162},
  year={2022},
  organization={Springer}
}

@inproceedings{liu2022dimension,
  title={Dimension reduction for efficient dense retrieval via conditional autoencoder},
  author={Liu, Zhenghao and Zhang, Han and Xiong, Chenyan and Liu, Zhiyuan and Gu, Yu and Li, Xiaohua},
  booktitle={Proceedings of the 2022 Conference on Empirical Methods in Natural Language Processing},
  pages={5692--5698},
  year={2022}
}

@inproceedings{zhang2024evaluating,
  title={Evaluating unsupervised dimensionality reduction methods for pretrained sentence embeddings},
  author={Zhang, Gaifan and Zhou, Yi and Bollegala, Danushka},
  booktitle={Proceedings of the 2024 Joint International Conference on Computational Linguistics, Language Resources and Evaluation (LREC-COLING 2024)},
  pages={6530--6543},
  year={2024}
}

@inproceedings{wang2023dimensionality,
  title={On the dimensionality of sentence embeddings},
  author={Wang, Hongwei and Zhang, Hongming and Yu, Dong},
  booktitle={Findings of the Association for Computational Linguistics: EMNLP 2023},
  pages={10344--10354},
  year={2023}
}

@inproceedings{takeshita2025randomly,
  title={Randomly removing 50\% of dimensions in text embeddings has minimal impact on retrieval and classification tasks},
  author={Takeshita, Sotaro and Takeshita, Yurina and Ruffinelli, Daniel and Ponzetto, Simone Paolo},
  booktitle={Proceedings of the 2025 Conference on Empirical Methods in Natural Language Processing},
  pages={27693--27714},
  year={2025}
}

@inproceedings{caspari2026corect,
  title={CoRECT: A framework for evaluating embedding compression techniques at scale},
  author={Caspari, Laura and Dinzinger, Michael and Dastidar, Kanishka Ghosh and Fellicious, Christofer and Mitrovi{\'c}, Jelena and Granitzer, Michael},
  booktitle={European Conference on Information Retrieval},
  pages={383--398},
  year={2026},
  organization={Springer}
}

@inproceedings{mhammedi2017efficient,
  title={Efficient orthogonal parametrisation of recurrent neural networks using householder reflections},
  author={Mhammedi, Zakaria and Hellicar, Andrew and Rahman, Ashfaqur and Bailey, James},
  booktitle={International Conference on Machine Learning},
  pages={2401--2409},
  year={2017},
  organization={PMLR}
}

@article{schreiber1989storage,
  title={A storage-efficient WY representation for products of Householder transformations},
  author={Schreiber, Robert and Van Loan, Charles},
  journal={SIAM Journal on Scientific and Statistical Computing},
  volume={10},
  number={1},
  pages={53--57},
  year={1989},
  publisher={SIAM}
}

@inproceedings{cronentownsend2002predicting,
  title={Predicting query performance},
  author={Cronen-Townsend, Steve and Zhou, Yun and Croft, W Bruce},
  booktitle={Proceedings of the 25th annual international ACM SIGIR conference on Research and development in information retrieval},
  pages={299--306},
  year={2002}
}

@article{cummins2014document,
  title={Document score distribution models for query performance inference and prediction},
  author={Cummins, Ronan},
  journal={ACM Transactions on Information Systems (TOIS)},
  volume={32},
  number={1},
  pages={1--28},
  year={2014},
  publisher={ACM New York, NY, USA}
}

@inproceedings{faggioli2023query,
  title={Query performance prediction for neural ir: Are we there yet?},
  author={Faggioli, Guglielmo and Formal, Thibault and Marchesin, Stefano and Clinchant, St{\'e}phane and Ferro, Nicola and Piwowarski, Benjamin},
  booktitle={European Conference on Information Retrieval},
  pages={232--248},
  year={2023},
  organization={Springer}
}

@inproceedings{wang2011cascade,
  title={A cascade ranking model for efficient ranked retrieval},
  author={Wang, Lidan and Lin, Jimmy and Metzler, Donald},
  booktitle={Proceedings of the 34th international ACM SIGIR conference on Research and development in Information Retrieval},
  pages={105--114},
  year={2011}
}

@article{deveaud2019learning,
  title={Learning to adaptively rank document retrieval system configurations},
  author={Deveaud, Romain and Mothe, Josiane and Ullah, Md Zia and Nie, Jian-Yun},
  journal={ACM Transactions on Information Systems (TOIS)},
  volume={37},
  number={1},
  pages={1--41},
  year={2018},
  publisher={ACM New York, NY, USA}
}

@inproceedings{ganguly2023query,
  title={Query-specific variable depth pooling via query performance prediction},
  author={Ganguly, Debasis and Yilmaz, Emine},
  booktitle={Proceedings of the 46th International ACM SIGIR Conference on Research and Development in Information Retrieval},
  pages={2303--2307},
  year={2023}
}

@inproceedings{scifact,
  title={Fact or Fiction: Verifying Scientific Claims},
  author={David Wadden and Shanchuan Lin and Kyle Lo and Lucy Lu Wang and Madeleine van Zuylen and Arman Cohan and Hannaneh Hajishirzi},
  booktitle={EMNLP},
  year={2020},
}

@inproceedings{nfcorpus,
  author = {Boteva, Vera and Gholipour, Demian and Sokolov, Artem and Riezler, Stefan},
  title = {A Full-Text Learning to Rank Dataset for Medical Information Retrieval},
  journal = {Proceedings of the 38th European Conference on Information Retrieval},
  journal-abbrev = {ECIR},
  year = {2016},
  city = {Padova},
  country = {Italy},
  url = {http://www.cl.uni-heidelberg.de/~riezler/publications/papers/ECIR2016.pdf}
}

@article{zhuang2024starbucks,
  title={Starbucks-v2: Improved training for 2d matryoshka embeddings},
  author={Zhuang, Shengyao and Wang, Shuai and Zheng, Fabio and Koopman, Bevan and Zuccon, Guido},
  journal={arXiv preprint arXiv:2410.13230},
  year={2024}
}

@article{li20242d,
  title={2d matryoshka sentence embeddings},
  author={Li, Xianming and Li, Zongxi and Li, Jing and Xie, Haoran and Li, Qing},
  journal={arXiv preprint arXiv:2402.14776},
  year={2024}
}

@inproceedings{huang2020embedding,
  title={Embedding-based retrieval in facebook search},
  author={Huang, Jui-Ting and Sharma, Ashish and Sun, Shuying and Xia, Li and Zhang, David and Pronin, Philip and Padmanabhan, Janani and Ottaviano, Giuseppe and Yang, Linjun},
  booktitle={Proceedings of the 26th ACM SIGKDD International Conference on Knowledge Discovery \& Data Mining},
  pages={2553--2561},
  year={2020}
}

@inproceedings{nigam2019semantic,
  title={Semantic product search},
  author={Nigam, Priyanka and Song, Yiwei and Mohan, Vijai and Lakshman, Vihan and Ding, Weitian and Shingavi, Ankit and Teo, Choon Hui and Gu, Hao and Yin, Bing},
  booktitle={Proceedings of the 25th ACM SIGKDD International Conference on Knowledge Discovery \& Data Mining},
  pages={2876--2885},
  year={2019}
}

@inproceedings{li2021embedding,
  title={Embedding-based product retrieval in taobao search},
  author={Li, Sen and Lv, Fuyu and Jin, Taiwei and Lin, Guli and Yang, Keping and Zeng, Xiaoyi and Wu, Xiao-Ming and Ma, Qianli},
  booktitle={Proceedings of the 27th ACM SIGKDD Conference on Knowledge Discovery \& Data Mining},
  pages={3181--3189},
  year={2021}
}

@inproceedings{covington2016deep,
  title={Deep neural networks for youtube recommendations},
  author={Covington, Paul and Adams, Jay and Sargin, Emre},
  booktitle={Proceedings of the 10th ACM conference on recommender systems},
  pages={191--198},
  year={2016}
}

@article{lewis2020retrieval,
  title={Retrieval-augmented generation for knowledge-intensive nlp tasks},
  author={Lewis, Patrick and Perez, Ethan and Piktus, Aleksandra and Petroni, Fabio and Karpukhin, Vladimir and Goyal, Naman and K{\"u}ttler, Heinrich and Lewis, Mike and Yih, Wen-tau and Rockt{\"a}schel, Tim and others},
  journal={Advances in neural information processing systems},
  volume={33},
  pages={9459--9474},
  year={2020}
}

@article{gao2023survey,
  title={Retrieval-augmented generation for large language models: A survey},
  author={Gao, Yunfan and Xiong, Yun and Gao, Xinyu and Jia, Kangxiang and Pan, Jinliu and Bi, Yuxi and Dai, Yi and Sun, Jiawei and Wang, Meng and Wang, Haofen},
  journal={arXiv preprint arXiv:2312.10997},
  year={2023}
}

@article{malkov2018efficient,
  title={Efficient and robust approximate nearest neighbor search using hierarchical navigable small world graphs},
  author={Malkov, Yu A and Yashunin, Dmitry A},
  journal={IEEE transactions on pattern analysis and machine intelligence},
  volume={42},
  number={4},
  pages={824--836},
  year={2018},
  publisher={IEEE}
}

@article{johnson2019billion,
  title={Billion-scale similarity search with GPUs},
  author={Johnson, Jeff and Douze, Matthijs and J{\'e}gou, Herv{\'e}},
  journal={IEEE transactions on big data},
  volume={7},
  number={3},
  pages={535--547},
  year={2019},
  publisher={IEEE}
}

@inproceedings{guo2020accelerating,
  title={Accelerating large-scale inference with anisotropic vector quantization},
  author={Guo, Ruiqi and Sun, Philip and Lindgren, Erik and Geng, Quan and Simcha, David and Chern, Felix and Kumar, Sanjiv},
  booktitle={International Conference on Machine Learning},
  pages={3887--3896},
  year={2020},
  organization={PMLR}
}

@inproceedings{ethayarajh2019contextual,
  title={How contextual are contextualized word representations? Comparing the geometry of BERT, ELMo, and GPT-2 embeddings},
  author={Ethayarajh, Kawin},
  booktitle={Proceedings of the 2019 conference on empirical methods in natural language processing and the 9th international joint conference on natural language processing (EMNLP-IJCNLP)},
  pages={55--65},
  year={2019}
}

@inproceedings{timkey2021rogue,
  title={All bark and no bite: Rogue dimensions in transformer language models obscure representational quality},
  author={Timkey, William and Van Schijndel, Marten},
  booktitle={Proceedings of the 2021 Conference on Empirical Methods in Natural Language Processing},
  pages={4527--4546},
  year={2021}
}

@article{mu2018allbutthetop,
  title={All-but-the-top: Simple and effective postprocessing for word representations},
  author={Mu, Jiaqi and Bhat, Suma and Viswanath, Pramod},
  journal={arXiv preprint arXiv:1702.01417},
  year={2017}
}

@article{kataiwa2025intrinsic,
  title={Measuring intrinsic dimension of token embeddings},
  author={Kataiwa, Takuya and Hakaze, Cho and Ohki, Tetsushi},
  journal={arXiv preprint arXiv:2503.02142},
  year={2025}
}

@inproceedings{inkiriwang2025dimensions,
  title={Do We Really Need All Those Dimensions? An Intrinsic Evaluation Framework for Compressed Embeddings.},
  author={Inkiriwang, Nathan and B{\"o}l{\"u}c{\"u}, Necva and Tarr, Garth and Rybinski, Maciej},
  booktitle={EMNLP (Findings)},
  pages={13305--13323},
  year={2025}
}

@inproceedings{raunak2019effective,
  title={Effective dimensionality reduction for word embeddings},
  author={Raunak, Vikas and Gupta, Vivek and Metze, Florian},
  booktitle={Proceedings of the 4th Workshop on Representation Learning for NLP (RepL4NLP-2019)},
  pages={235--243},
  year={2019}
}

@article{yang2026,
  title={Slipstream: Locality-Aware Graph Index Construction for Streaming Approximate Nearest Neighbor Search},
  author={Yang, Shubing and Zhao, Dongfang},
  journal={arXiv preprint arXiv:2606.02992},
  year={2026}
}

@book{draper1998applied,
  title={Applied regression analysis},
  author={Draper, Norman R and Smith, Harry},
  volume={326},
  year={1998},
  publisher={John Wiley \& Sons}
}

@book{montgomery2022introduction,
  title={Introduction to Linear Regression Analysis, 6e Solutions Manual},
  author={Montgomery, Douglas C and Peck, Elizabeth A and Vining, G Geoffrey},
  year={2022},
  publisher={John Wiley \& Sons}
}

@article{chicco2021coefficient,
  title={The coefficient of determination R-squared is more informative than SMAPE, MAE, MAPE, MSE and RMSE in regression analysis evaluation},
  author={Chicco, Davide and Warrens, Matthijs J and Jurman, Giuseppe},
  journal={Peerj computer science},
  volume={7},
  pages={e623},
  year={2021},
  publisher={PeerJ Inc.}
}

% \appendix





\end{document}